\documentclass[runningheads]{llncs}
\usepackage{graphicx}
\usepackage{mathtools}
\usepackage{enumitem}
\usepackage{tabularx}
\usepackage{todonotes}
\usepackage{amsmath}

\usepackage[english]{babel}
\usepackage{amsthm}
\usepackage{amssymb}
\theoremstyle{definition}

\usepackage{listings}
\usepackage{algorithm} 
\usepackage{algpseudocode} 
\usepackage{caption}
\usepackage{subcaption}
\usepackage{cite}
\usepackage{booktabs}

\usepackage[edges]{forest}

\usepackage{hyperref}

\begin{document}
\title{Classifying sequences by combining context-free grammars and OWL ontologies}
\titlerunning{Classifying sequences with CFG and OWL}
%
\author{Nicolas Lazzari\inst{1, 3}\orcidID{0000-0002-1601-7689} \and
Andrea Poltronieri\inst{2}\orcidID{0000-0003-3848-7574} \and
Valentina Presutti\inst{1}\thanks{Alphabetical order}\orcidID{0000-0002-9380-5160}}
\authorrunning{N. Lazzari et al.}
%
\institute{
LILEC, University of Bologna, Italy \\
\email{nicolas.lazzari2@studio.unibo.it, valentina.presutti@unibo.it} \and
Department of Computer Science and Engineering, University of Bologna, Italy \\ 
\email{andrea.poltronieri2@unibo.it}  \and
Department of Computer Science, University of Pisa, Italy
}
\maketitle              
\begin{abstract}
This paper describes a pattern to formalise context-free grammars in OWL and its use  for sequence classification. The proposed approach is compared to existing methods in terms of computational complexity as well as pragmatic applicability, with examples in the music domain. 
\end{abstract}

\section{Introduction} \label{sec:introduction}
The introduction of formal grammars by Chomsky in the 50s \cite{chomsky1956cfg}, and in particular Context Free Grammars (CFG), led to prolific research in the area of Natural Language Processing. Methods based on statistical language modeling have mostly replaced formal grammars, nevertheless research in this area is still relevant as many domains and tasks benefit from their application. For example, an important application of formal grammars concerns high level programming languages. Through an efficient parsing process \cite{aho1974lr}, machine-level instructions can be abstracted in human readable instructions. Theoretically, every problem that can be abstracted as a \emph{sequence of symbols} can be modeled with formal grammars, which  makes them a suitable tool for \emph{sequence classification}, the task we focus on in this paper. In the biology field, the classification of RNA secondary structures has been performed using CFG \cite{damasevicius2010dna, rivas2000rna, giegerich2014rna}. Similarly, in the music field, CFG are used to classify different types of harmonic and melodic sequences \cite{rohrmeier2011cfg, tidhar2005grammar, steedman1984jazz, keller2013improviser, baroni1984melody, tojo2006hpsg}.

\textbf{Background.} A \emph{language} is a collection of sequences, each defined according to a finite set of symbols \cite{chomsky1959cfgproperties}. A \emph{grammar} can be interpreted as a function of a language, having a set of symbols as its domain, and a set of sequences as its range. 
Defining a formal grammar is a complex task that requires deep knowledge of the application domain as well as good modeling skills, to obtain an efficiently parsable grammar. CFGs can be parsed in less than $O(n^3)$ \cite{valiant1975cfgcompl}, with $n$ the length of the string. However when symbols have ambiguous semantics and require additional attributes to be disambiguated, parsing a sequence is NP-complete \cite{barton1987nlpcompl}.
Ambiguous symbols are a common issue, for instance in the case of polysemous words in natural language or diminished chords in music. This problem can be mitigated through the use of complex notations, such as SMILES~\cite{weininger1988smiles} to represent molecules in the biology field or Harte~\cite{harte2005symbolic} to represent musical chords. Nevertheless, these notations are either hard to interpret, requiring additional tools to be converted back into a human understandable format, or they cause loss of information.

We address the problem of sequence classification by proposing a hybrid approach that combines the use of Context Free Grammar (CFG) parsers with OWL ontologies. We define a pattern to formalise CFG by providing a novel definition for CFG based on Description Logic. We define a set of algorithms to produce an OWL ontology based on this pattern that supports the alignment of symbols in a sequence to its classes. Our approach is based on the identification of sub-sequences according to the taxonomy defined in the ontology. We argue that our proposal has a relevant pragmatic potential as it enables sequence classification based on semantic web knowledge representation, therefore supporting the linking of Context Free Grammars to web ontologies and knowledge graphs.

The contribution of this research can be summarised as follows:
\begin{itemize}
    \item defining a novel formalisation of Context Free Grammars based on Description Logic;
    \item providing an algorithm for the conversion of such formalisation in OWL;
    \item demonstrating the correctness, computational complexity and applicability of the proposed method in the music domain.
\end{itemize}

The paper is organized as follows: in Section \ref{sec:related}, an overview of related works on sequence classification is presented. Section \ref{sec:preliminaries} provides relevant definitions for Context Free Grammars that are used later in Section \ref{sec:method}, to describe the formalisation of CFGs in DL. Section \ref{sec:method:cfgdl-in-owl} describes our approach to sequence classification. In section \ref{sec:experiments} we evaluate our method on the task of sequence classification in the music domain. Finally, in Section \ref{sec:conclusions}, we summarize the contribution and discuss future development.


\section{Related Work}  \label{sec:related}
Relevant work to this contribution include: techniques for sequence classification, sequence model with grammars, approaches to integrate CFG and semantic web technologies, and their application in the music domain (cf. Section \ref{sec:experiments}).
\\

\noindent \textbf{Sequence Classification} (SQ)
is the task of predicting the class of an input, defined as a sequence over time or space, among a predefined set of classes \cite{madhusudhan2019indianmusic}.
SQ is relevant in several application fields, such as genomics research \cite{sanjeev2021gene}, health informatics, abnormality detection and information retrieval \cite{xing2010survey}, Natural Language Processing (NLP) \cite{lei2019nlp, chen2016sentiment}. 
In \cite{xing2010survey} different methodologies for SQ are identified, such as feature-based classification, sequence distance-based classification and support vector machines (SVM).
The most advanced approaches mainly rely on Deep Learning (DL) \cite{lobosco2017dl, carrasco2019dl}. 
SQ is relevant in the music domain, as sequences are at its core, for instance melodic and harmonic sequences that span over a temporal dimension.
An example of sequence classification applied to music 
is \cite{madhusudhan2019indianmusic}, which addresses the recognition of \textit{raga}
using recurrent neural networks (LSTM-RNN).\\

\noindent\textbf{Sequence Model with Grammars.}
There is a close relationship between sequences and formal grammars. A classic related task (ranging from natural language processing \cite{lawrence2000nlp} to bio-informatics \cite{damasevicius2010dna}) is \textit{grammar inference}~\cite{young-lai2009inference}. 
Grammars are used for the classification of sequences of different types, mainly for analysing genetic sequences \cite{rivas2000rna,giegerich2014rna}.
There are applications of grammars for  music classification. The most renowned example is the generative theory of tonal music (GTTM) \cite{lerdahl1983gttm}, analogous to Chomsky's transformational or generative grammar. \cite{chomsky1956cfg}. 
Although GTTM does not explicitly provide generative grammar rules, this work has inspired the formalisation of a wide range of context-free rules, describing different music genres \cite{tidhar2005grammar, steedman1984jazz, keller2013improviser}, melody \cite{baroni1984melody} and harmony \cite{rohrmeier2011cfg, tojo2006hpsg}.
These works are relevant input to our work as they formalise aspects of certain types of sequences into rules. 
\\

\noindent \textbf{Sequence Model with OWL.} There are proposals to use OWL to classify sequences. 
However, one of the main challenges when dealing with sequences in OWL is to organise the elements being described in an ordered fashion, as proposed in \cite{drummond2006order}. For instance, OWL reasoning is employed for classifying genomic data \cite{wolstencroft2007genomic}.
A method for analysing jazz chord sequences is proposed in \cite{pachet2000solarblues}. This system is based on an ontology which, through reasoning, produces a hierarchical jazz sequence analysis. 
Similarly, two OWL ontologies, $\mathcal{MEO}$ and $\mathcal{SEQ}$, are presented in \cite{wissmann2012chord} that combined with a CFG parser support sequence classification. Nevertheless, this method is only able to represent \textit{safely-concatenable} CFG, while we overcome this limitation in our approach (cf. Section \ref{sec:comparison}).

\section{Preliminaries}\label{sec:preliminaries}
This section introduces the notation and the definitions used in Section \ref{sec:method}, based on \cite{hopcroft2001automata} that the reader can consult for details.

\begin{definition}[Context Free Grammar]
\label{def:cfg}
A Context Free Grammar (CFG) $G = (V, \Sigma, R, S)$ consists of a finite set of non-terminal $V$ (variables), a set of terminals $\Sigma$ such that $\Sigma \cap V = \emptyset$, a set of functions $R \subseteq V \times (V \times R)^*$ (production), and a starting symbol $S \in V$.
\end{definition}

\begin{definition}[Language of a grammar]
\label{def:l-cfg}
The language of a grammar $G(V, \Sigma, R, S)$ is defined as $L(G) = \{ w \in \Sigma^*: S \xRightarrow{*} w \} $, where $S \xRightarrow{*} w$ represents the consecutive application of production $f \in R$ starting from the initial symbol $S$, called derivation.
\end{definition}

\begin{example}[Context Free Grammar]
\label{ex:cfg}
Let $G = (V, \Sigma, R, S)$ with
\begin{alignat*}{3}
    &V = \{ &&\text{Expression}, \text{Bit} \} \\
    &\Sigma = \{ &&0, 1, + \} \\
    &R = \{ &&\text{Expression} \rightarrow \text{Expression} + \text{Expression}\ |\ \text{Bit}\ 0\ |\ \text{Bit}\ 1\ |\  0\ |\  1, \\
    & &&\text{Bit} \rightarrow \text{Bit}\ 0\ |\ \text{Bit}\ 1\ |\ 0\ |\ 1\ \} \\
    &S = &&\text{Expression}
\end{alignat*}
where $X \rightarrow X_1 | \cdots | X_n$ is a shorthand for $\{ X \rightarrow X_1, \ \cdots\ X \rightarrow X_n \}$
\end{example}

Example \ref{ex:cfg} shows a simple grammar used to parse the sum of two binary numbers. Its language $L$, as defined in definition \ref{def:l-cfg}, is of the form $L = \{ 0+0, 0+1, 1+0, 10+0, \cdots, 11010+10, \cdots \}$.

In order to express a concise and effective conversion method and its corresponding proof we only consider grammars in Chomsky Normal Form, as defined in Definition \ref{def:cnf}. This results in homogeneous productions in the form \begin{align*}
    A \rightarrow B C \quad \text{or} \quad A \rightarrow t
\end{align*}

where $A, B, C \in V$ and $t \in \Sigma$.

\begin{definition}[Chomsky Normal Form]
\label{def:cnf}
A Context Free Grammar is in Chomsky Normal Form (CNF) if the set of functions $R \subseteq V \times (V \backslash \{S\} \times R)^2$
\end{definition}

Note that the imposed restriction does not imply any loss in expressiveness, since any context-free grammar can be converted in CNF \cite{hopcroft2001automata}. For instance, Example \ref{ex:cfg}, converted in CNF, results in the grammar in Example \ref{ex:cfg-cnf}.

\begin{example}[Example \ref{ex:cfg} in Chomsky Normal Form]
\label{ex:cfg-cnf}
Let $G = (V, \Sigma, R, S)$ with
\begin{alignat*}{3}
    &V = \{ &&\text{Expression}, \text{Expression}_0, \text{Bit}, \text{Zero}, \text{One}, \text{Plus} \} \\
    &\Sigma = \{ &&0, 1, + \} \\
    &R = \{ &&\text{Expression} \rightarrow \text{Expression}_0\ \text{Expression}\ |\ \text{Bit}\ \text{Zero}\ |\ \text{Bit}\ \text{One}\ |\ 0\ |\ 1, \\
    & &&\text{Expression}_0 \rightarrow \text{Expression}\ \text{Plus}, \\
    & &&\text{Bit} \rightarrow \text{Bit}\ \text{Zero}\ |\ \text{Bit}\ \text{One}\ |\ 0\ |\ 1, \\
    & &&\text{Plus} \rightarrow +, \\
    & &&\text{Zero} \rightarrow 0, \\
    & &&\text{One} \rightarrow 1 \}
\end{alignat*}
\end{example}

When parsing a language based on a grammar, it is useful to visualise the derivation process as a parse tree. 

\begin{definition}[Parse tree of a CFG]
\label{def:cfg-parse-tree}
A parse tree $T$ of $G = (V, \Sigma, R, S)$ is a tree in which each leaf $l \in (V \cup \Sigma)$ and each inner node $n_i \in V$.
Given $c_1 \cdots c_n$ the children of an inner node $n_i$ then $\exists f \in R$ s.t. $f: n_i \rightarrow c_1 \cdots c_n$.\cite{hopcroft2001automata}    
\end{definition}

\begin{corollary}
\label{def:cnf-binary-parse-tree}
Given $T$ the parse tree of a CFG in CNF $\Rightarrow$ $T$ is a binary tree.
\end{corollary}

Corollary \ref{def:cnf-binary-parse-tree} follows from Definitions \ref{def:cnf} and \ref{def:cfg-parse-tree}, since each production is either a unary or a binary function. Figure \ref{fig:ex-parse-tree-cfg} shows  the parse tree of the sequence \textit{1+0} from the grammar defined by Example \ref{ex:cfg} and Figure \ref{fig:ex-parse-tree-cnf} shows the parse tree the grammar defined by Example \ref{ex:cfg-cnf}.

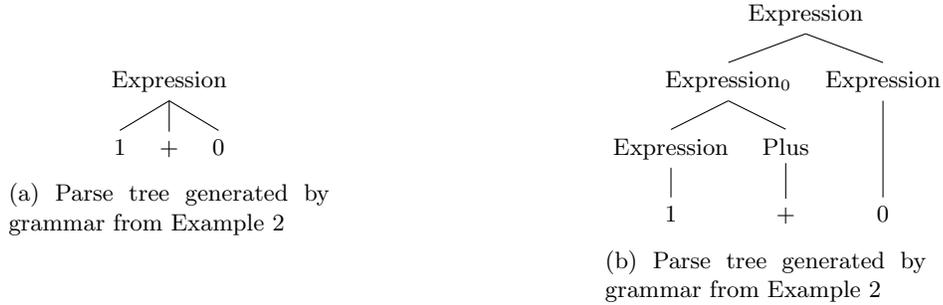
\begin{figure}[ht]
    \centering
    \begin{subfigure}{0.35\textwidth}
        \centering
        \begin{forest}
        for tree={parent anchor=south, child anchor=north}
        [Expression [1] [+] [0] ]
        \end{forest}
        \caption{Parse tree generated by grammar from Example \ref{ex:cfg-cnf}}
        \label{fig:ex-parse-tree-cfg}
    \end{subfigure}
    \hfill
    \begin{subfigure}{0.35\textwidth}
        \centering
        \begin{forest}
        for tree={parent anchor=south, child anchor=north}
        [Expression
            [Expression$_0$ [Expression [1, tier=word]] [Plus [+, tier=word]]]
            [Expression [0, tier=word]]]
        \end{forest}
        \caption{Parse tree generated by grammar from Example \ref{ex:cfg-cnf}}
        \label{fig:ex-parse-tree-cnf}
    \end{subfigure}
    \caption{Parse trees obtained from the sequence \textit{1+0}}
    \label{fig:ex-parse-tree}
\end{figure}

\section{Formalising Context-Free Grammars using Description Logic}\label{sec:method}

As OWL is based on Description Logic (DL) theory, we define a DL-based formalisation of Context-Free Grammars (in Chomsky Normal Form), which we refer to as CFG-DL.
CFG-DL is based on Definition \ref{def:cfg}, where variables and terminals are represented as concepts\footnote{DL concepts translate into OWL classes}.
We demonstrate that any CFG can be converted in a CFG-DL (cf. Theorem \ref{theorem:cfg_in_cfgdl}) and that such conversion can be performed in $O(n)$ (cf. Theorem \ref{theorem:cfg_in_cfgdl_linear}).
Theorem \ref{theorem:cfg_in_cfgdl} and its proof rely on the concept of \emph{rolification}, which formalises axioms that act as rules in the form \emph{if-then} \cite{krisnadhi2011owlandrules}. 
For each concept $C$ a corresponding axiom $R_C$ is created and the restriction $C \equiv R_C . Self$ is imposed. By chaining together different axioms it is possible to define \emph{if-then} rules. 
For a more in-depth explanation, please refer to Krisnadhi et al. \cite{krisnadhi2011owlandrules}.

\begin{definition}[CFG-DL]
\label{def:cfgdl}
A CFG-DL $G_{DL} = (C_v, R, C_\Sigma, S)$ consists of a finite set of concepts $C_v$, a finite set of concepts $C_\Sigma$, a set of axioms $R$, and a starting concept $S \in C_v$.
\end{definition}

\begin{theorem}
\label{theorem:cfg_in_cfgdl}
Every Context-Free Grammar $G$ in Chomsky Normal Form can be converted in a CFG-DL $G_{DL}$.
\end{theorem}

\begin{proof}
\label{proof:cfg_in_cfgdl}
Given a Context-Free grammar $G = (V, \Sigma, R, S)$ in Chomsky Normal Form we can obtain the corresponding $G_{DL} = (C'_v, R', C'_\Sigma, S')$ as follows:
\begin{enumerate}
    \item \label{proof:cfg_in_cfgdl:1}
        $\forall v \in V$ let $C_v$ be a concept such that $C_v \sqsubseteq C'_v$ \\
        $\xRightarrow \quad \forall v \in V \  \exists C_v \sqsubseteq C'_v$, where $C_v$ is the respective concept of the variable $V$.

    \vspace{0.5em}

    \item \label{proof:cfg_in_cfgdl:2}
        $\forall t \in \Sigma$ let $C_t$ be a concept such that $C_t \sqsubseteq C'_\Sigma$ \\
        $\xRightarrow \quad \forall t \in \Sigma \  \exists C_t \sqsubseteq C'_\Sigma$, where $C_t$ is the respective concept of the terminal $t$.

    \item \label{proof:cfg_in_cfgdl:3}
          Let $f \in R$. It follows from Definition \ref{def:cnf} that $f$ is of either type:
          \begin{enumerate}
              \item $R \rightarrow A B$ such that $R \in V$, $A, B \in V \cup \Sigma$. 
              \item $R \rightarrow t$ such that $R \in V$ and $t \in \Sigma$;
          \end{enumerate}
          Both cases can respectively be represented in DL as follows:
          \begin{enumerate}
              \item \label{proof:cfg_in_cfgdl:2:1}
                    \begin{enumerate}
                        \item \label{proof:cfg_in_cfgdl:2:1:1}
                              Let $C_R \sqsubseteq C'_v, C_A \sqsubseteq C'_v, C_B \sqsubseteq C'_v$ be the respective concepts of $R, A, B$ definied in step \ref{proof:cfg_in_cfgdl:1}
                         \item \label{proof:cfg_in_cfgdl:2:1:2}
                                Let $R_R, R_A, R_B$ be the rolification \cite{krisnadhi2011owlandrules} of the concepts $C_R, C_A, C_B$ such that $C_R \equiv \exists R_R.Self$, $C_A \equiv \exists R_A.Self$, and $C_B \equiv \exists R_B.Self$.
                         \item \label{proof:cfg_in_cfgdl:2:1:3}
                                Let $R_{next}$ be the role such that $C_1 \circ R_{next} \circ C_2$ has semantic meaning \textit{$C_1$ has as next element in the sequence $C_2$}, with $C_1 \sqsubseteq C'$ and $C_2 \sqsubseteq C'$
                         \item \label{proof:cfg_in_cfgdl:2:1:4}
                                Let $V_1 \equiv \exists R_1.Self$ and $V_2 \equiv \exists R_2.Self$ be roles such that $R_A \circ R_{next} \circ R_B \sqsubseteq R_1$ and $R_B \circ R_{next}^{-1} \circ R_A \sqsubseteq R_2$. \\
                    \end{enumerate}
                    $\xRightarrow{\ref{proof:cfg_in_cfgdl:2:1:1}, \ref{proof:cfg_in_cfgdl:2:1:2}, \ref{proof:cfg_in_cfgdl:2:1:3}, \ref{proof:cfg_in_cfgdl:2:1:4}} \quad R \rightarrow A B \Longleftrightarrow (C_A \sqcap V_1) \sqcup (C_B \sqcap V_2) \sqsubseteq C_R$.
              \item \label{proof:cfg_in_cfgdl:2:2}
                    Let $C_t \sqsubseteq C'_\Sigma$ be the concept of the terminal $t$ defined in step \ref{proof:cfg_in_cfgdl:2} and $C_R \sqsubseteq C'_V$ be the concept of variable $R$ defined in step \ref{proof:cfg_in_cfgdl:1} \\
                    $\Rightarrow R \rightarrow t \Longleftrightarrow C_t \sqsubseteq C_r $
          \end{enumerate}
          $\xRightarrow{\ref{proof:cfg_in_cfgdl:2:1}, \ref{proof:cfg_in_cfgdl:2:2}} \quad f \in R', \ \forall f \in R$ .

    \item \label{proof:cfg_in_cfgdl:4}
          $\xRightarrow{\ref{proof:cfg_in_cfgdl:1}} \exists C_s \sqsubseteq C'_V$ where $C_s$ is the concept corresponding to $S$, as defined in step \ref{proof:cfg_in_cfgdl:1}.
\end{enumerate}

$\xRightarrow{\ref{proof:cfg_in_cfgdl:1}, \ref{proof:cfg_in_cfgdl:2}, \ref{proof:cfg_in_cfgdl:3}, \ref{proof:cfg_in_cfgdl:4}} G_{DL} \equiv G$.
\end{proof}

\begin{theorem}
\label{theorem:cfg_in_cfgdl_linear}
The conversion between a CFG $G = (V, \Sigma, R, S)$ and a CFG-DL $G_{DL} = (C'_v, R', C'_\Sigma, S')$ can be performed in $O(n)$, in particular O($|V| + |\Sigma| + |R|$).
\end{theorem}

\begin{proof}
It follows from the proof of Theorem \ref{theorem:cfg_in_cfgdl} as we only need to loop through each element of $V$, $\Sigma$ and $R$ at most one time.
\end{proof}


We remark that in Definition \ref{def:cfgdl} terminals are modeled as concepts and stand at the same level of variables. At first, it might seem more intuitive to represent terminals as individuals. But this would radically change the semantic meaning of an element in a sequence. 
Take for example the sequence \textit{10+11} from the language of  grammar in Example \ref{ex:cfg}. There are three occurrences of terminal \textit{1}, but they are fundamentally different entities: the first occurrence of terminal \textit{1} is characterized by its syntactic aspect as well as its position with respect to the whole sequence. 
If we represent each terminal as an individual then each occurrence of that terminal in a sequence would be represented by the very same individual. This would invalidate the semantics of the whole sequence and yield a wrong formalization.
In order to address this issue, we need a proper definition of how to represent a sequence in description logic. We do that by adapting Definition \ref{def:l-cfg} to CFG-DL.

\begin{definition}[Language of a CFG-DL]
\label{def:l-cfgdl}
Let $G = (C_v, R, C_\Sigma, S)$ be a CFG-DL, we define as $L(G)$ the set of sequences $\overline{s}$ such that, given $N$ the number of elements in the sequence $\overline{s}$, $\overline{s} \equiv (C_1 \sqcap \exists R_{\text{next}}.C_2) \sqcap \cdots \sqcap (C_{N - 1} \sqcap \exists R_{\text{next}}.C_N)$, with $R_{\text{next}}$ the role defined in step \ref{proof:cfg_in_cfgdl:2:1:3} of Theorem \ref{theorem:cfg_in_cfgdl}'s proof and $C_t \sqsubseteq C_\Sigma, t \in [1, N]$.
\end{definition}

\section{Sequence classification using CFG-DL}
\label{sec:method:cfgdl-in-owl}
CFG-DL can be represented in OWL, as OWL2 direct semantics is based on Description Logic \cite{horrocks2012owl}. We devise an algorithm based on Definition \ref{def:cfgdl} and on the respective constructive proof \ref{proof:cfg_in_cfgdl} of Theorem \ref{theorem:cfg_in_cfgdl}. 
Algorithm \ref{algo:cfg-in-owl} converts a CFG to OWL, without generating any intermediary CFG-DL. A similar algorithm can be defined to convert a CFG in CFG-DL, following the constructive proof \ref{proof:cfg_in_cfgdl}. 

Triples are written in Manchester syntax \cite{horridge2006manchester}. We use the symbol $\blacktriangleright$ to indicate the OWL triples that need to be created. 
We generally use $\overline{R}_{next}$ as the role $R_{next}$ defined in step \ref{proof:cfg_in_cfgdl:2:1:3} of Thereom \ref{theorem:cfg_in_cfgdl}'s proof. Any arbitrary OWL property can be used as $\overline{R}_{next}$ as long as it is a functional property, such as the \textit{seq:directlyPrecedes} property from the \textit{sequence} Ontology Design Pattern \cite{hitzler2016odp}. 
Algorithm \ref{algo:cfg-in-owl} has also complexity $O(n)$: similarly to the considerations on Theorem \ref{theorem:cfg_in_cfgdl_linear}, we only need to loop through each element of $V$, $\Sigma$ and $R$ at most one time.

\begin{algorithm}[htp]
    \caption{CFG in OWL}
    \label{algo:cfg-in-owl}
    \begin{algorithmic}
        \Require $G = (V, \Sigma, R, S)$
        \State $\blacktriangleright$ \texttt{ObjectProperty: $\overline{R}_1$}   
        \State $\blacktriangleright$ \texttt{ObjectProperty: $\overline{R}_2$}   
        \State $\blacktriangleright$ \texttt{Class: $\overline{V}_1$ EquivalentTo: $\overline{R}_1$ some}   
        \State $\blacktriangleright$ \texttt{Class: $\overline{V}_2$ EquivalentTo: $\overline{R}_2$ some}
        
        \For{$\overline{v} \in V$}
            \State {$\blacktriangleright$ \texttt{ObjectProperty: $R_{\overline{v}}$}}
            \State {$\blacktriangleright$ \texttt{Class: $C_{\overline{v}}$ EquivalentTo: $R_{\overline{v}}$ some Self }}
        \EndFor

        \For{$\overline{t} \in \Sigma$}
            \State {$\blacktriangleright$ \texttt{ObjectProperty: $R_{\overline{t}}$}}
            \State {$\blacktriangleright$ \texttt{Class: $C_{\overline{t}}$ EquivalentTo: $R_{\overline{t}}$ some Self }}
        \EndFor

        \For {$r \in R$}
            \If{$r$ is of type $R \rightarrow A B$}
                \State \textit{with $\overline{C}_R$, $\overline{C}_A$, $\overline{C}_B$ being the respective concepts of $R, A, B$}
                \State \textit{with $\overline{R}_A$, $\overline{R}_B$ being the respective rolification of $A, B$}
                \State \texttt{$\blacktriangleright$ ObjectProperty: $\overline{R_1}$ SubPropertyChain: $\overline{R}_A$ o $\overline{R}_{next}$ o $\overline{R}_B$}
                \State \texttt{$\blacktriangleright$ ObjectProperty: $\overline{R}_2$ SubPropertyChain: $\overline{R}_B$ o inverse($\overline{R}_{next}$) o $\overline{R}_A$}
                \State $\blacktriangleright$ \texttt{($\overline{C}_A$ and $\overline{V}_1$) or ($\overline{C}_B$ and $\overline{V}_2$) SubClassOf: $\overline{C}_R$ }
            \ElsIf{$r$ is of type $R \rightarrow t$}
                \State \textit{with $\overline{C}_R$, $\overline{C}_t$ being the respective concepts of $R, t$}
                \State \texttt{$\blacktriangleright$ Class: $\overline{C}_t$ SubClassOf: $\overline{C}_R$}
            \EndIf
        \EndFor
    \end{algorithmic} 
\end{algorithm}

The rolification of the classes \textit{$\overline{V}_1$} and \textit{$\overline{V}_2$} is performed by using the existential restriction on \texttt{owl:Thing}. 
This prevents the creation of non-simple properties due to the use of property chain later in the algorithm and allows the usage of reasoners such as Hermit \cite{glimm2014hermit} or Pellet \cite{sirin2007pellet}.

Sequences must be converted to be used in the ontology obtained with Algorithm \ref{algo:cfg-in-owl}. Algorithm \ref{algo:seq-in-owl-for-cfgdl} presents an algorithm that performs such conversion in $O(n)$. 
It is based on Definition \ref{def:l-cfgdl}. Analogously to Algorithm \ref{algo:cfg-in-owl}, we express triples in Manchester syntax using the symbol $\blacktriangleright$ and we use $\overline{R}_{next}$ as the role $R_{next}$ defined in step \ref{proof:cfg_in_cfgdl:2:1:3} of Thereom \ref{theorem:cfg_in_cfgdl}'s proof.

\begin{algorithm}[htp]
    \caption{Sequence in OWL for CFG-DL}
    \label{algo:seq-in-owl-for-cfgdl}
    \begin{algorithmic}
        \Require $G = (V, \Sigma, R, S)$
        \Require $s \subseteq \Sigma^*$ the sequence to represent
        \Require $N$ the length of the sequence $s$

        \For{$i \in [1, N - 1]$}
            \State $s_i \gets s[i]$
            \State $s_n \gets s[i + 1]$
            \State \textit{with $\overline{C}_i$, $\overline{C}_n$ being the respective concepts of the terminals $s_i, s_n$}
            \State \texttt{$\blacktriangleright$ Individual: $s_n$ Types: $\overline{C}_n$}
            \State \texttt{$\blacktriangleright$ Individual: $s_i$ Types: $\overline{C}_i$ Facts: $\overline{R}_{next}\ s_n$ }
        \EndFor
    \end{algorithmic} 
\end{algorithm}

\begin{figure}[htp]
    \centering
    \includegraphics[width=0.85\textwidth]{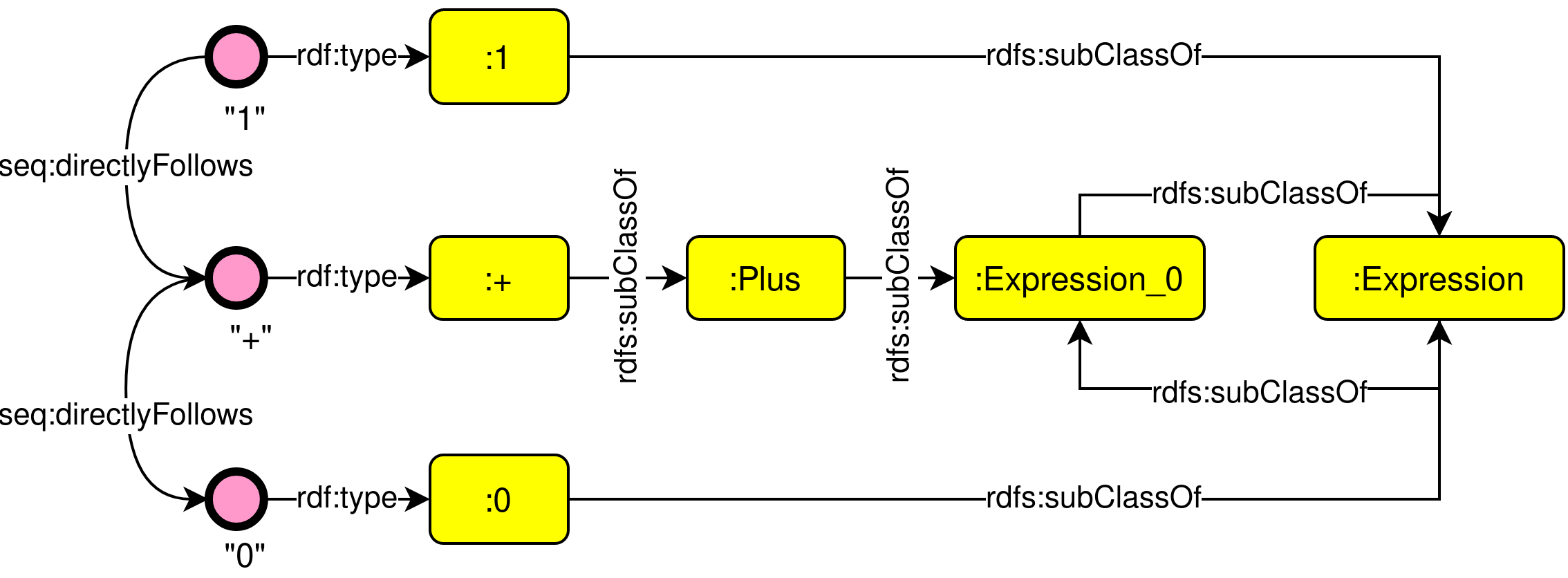}
    \caption{Sequence \textit{1+0} parsed by the grammar from example \ref{ex:cfg-cnf} represented as CFG-DL in OWL}
    \label{fig:example-cfgdl-owl}
\end{figure}

Figure \ref{fig:example-cfgdl-owl}, shows the grammar from Example \ref{ex:cfg-cnf}, converted to a CFG-DL in OWL with Algorithm \ref{algo:cfg-in-owl}, used to parse the sequence \textit{1 + 0}, converted using Algorithm \ref{algo:seq-in-owl-for-cfgdl}.
We can see how the whole sequence is correctly classified to be of class \texttt{Expression} and how \texttt{Expression\_0} and \texttt{Expression} are classified as  subclass of each other. Indeed, \texttt{Expression\_0} and \texttt{Expression} are equivalent. This can be observed from the normalization process performed on Example \ref{ex:cfg} that resulted in Example \ref{ex:cfg-cnf}: \textit{Expression\_0} variable is introduced to obtain a binary projection of $\textit{Expression} \rightarrow \textit{Expression} + \textit{Expression}$, as required by CNF.
If we substitute every occurrence of \textit{Expression\_0} with its right hand side (\textit{Expression} + \textit{Expression}) an equivalent grammar, which is not in CNF, is obtained.
The overall pattern in Figure \ref{fig:general-pattern} can be generalized to every CFG converted in OWL.

\begin{figure}[htp]
    \centering
    \begin{subfigure}{0.45\textwidth}
        \centering
        \includegraphics[width=\textwidth]{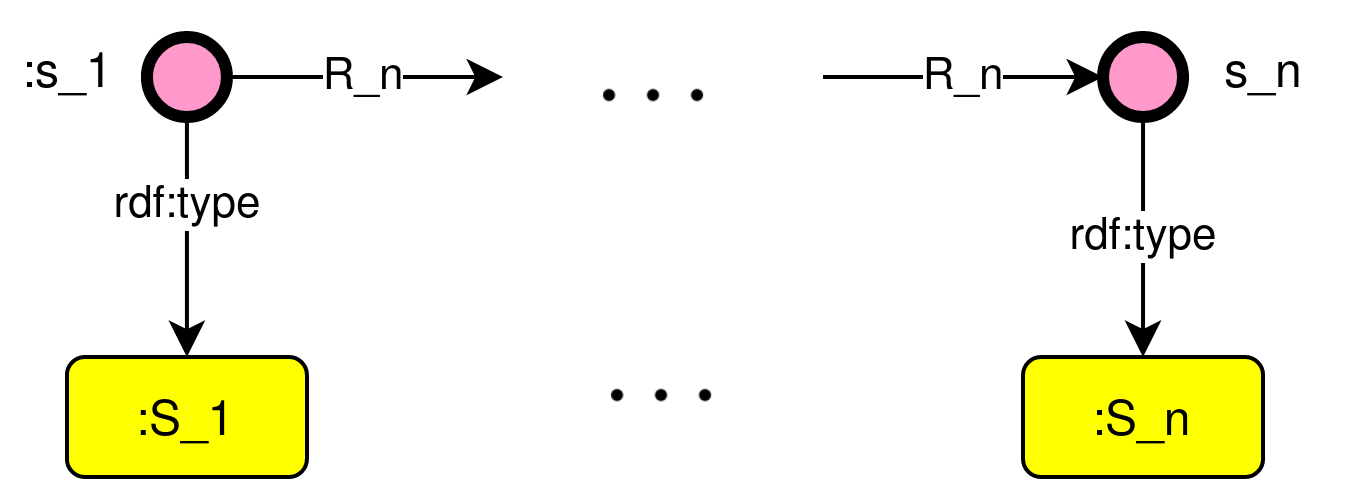}
        \caption{}
        \label{fig:pattern-seq-in-owl-for-cfgdl}
    \end{subfigure}
    \hfill
    \begin{subfigure}{0.54\textwidth}
        \centering
        \includegraphics[width=\textwidth]{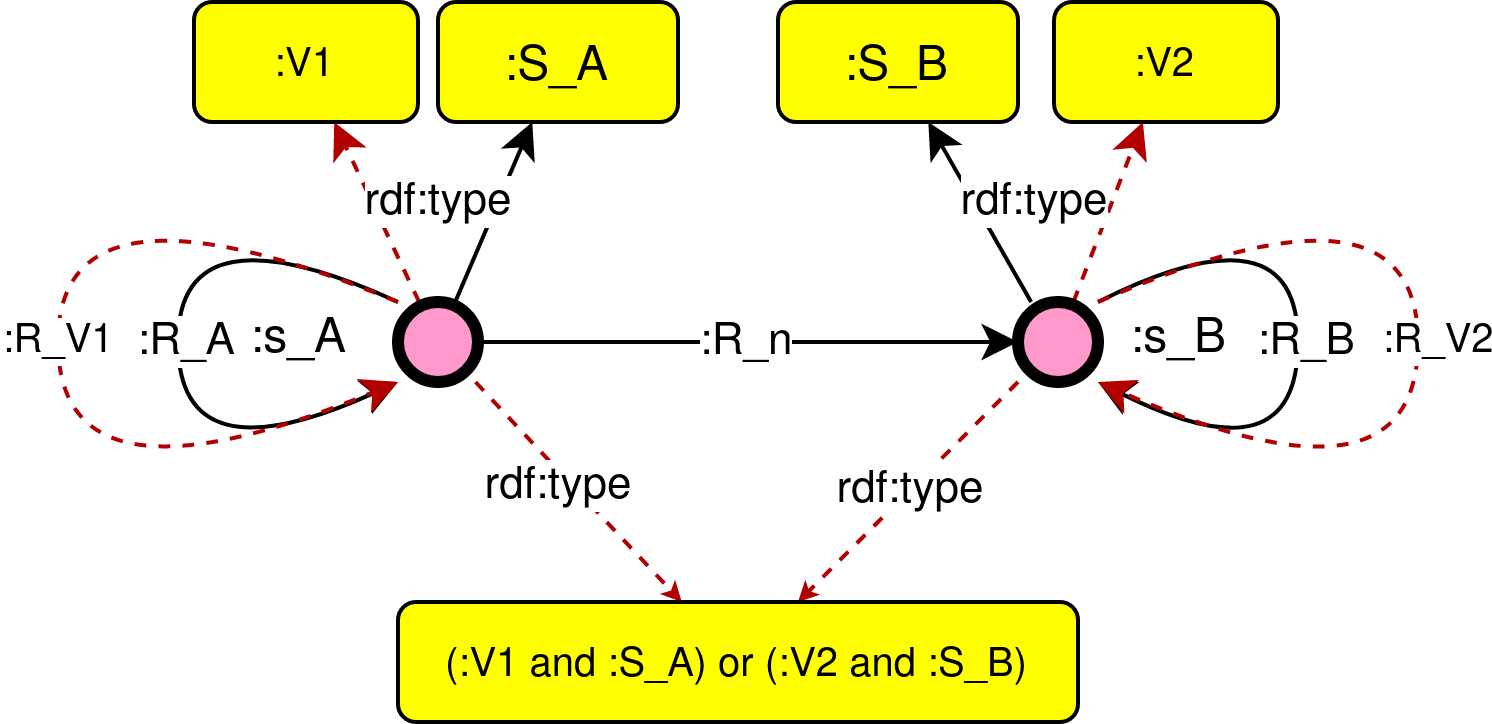}
        \caption{}
        \label{fig:cfgdl-reasoning}
    \end{subfigure}
    \caption{General ontology patterns from Algorithms \ref{algo:cfg-in-owl} (\ref{fig:pattern-seq-in-owl-for-cfgdl}) and \ref{algo:seq-in-owl-for-cfgdl} (\ref{fig:cfgdl-reasoning}), for a rule of type $R \rightarrow A\ B$. The red arrows are created by the reasoner, due to the definition of the property chain and the general axiom. The pattern from a rule of type $R \rightarrow t$ is a simple subsumption relation. }
    \label{fig:general-pattern}
\end{figure}

\subsection{Computational complexity} \label{sec:method:complexity}
An ontology produced by Algorithm \ref{algo:cfg-in-owl} is in DL $\mathcal{SROIEL}$, which is contained in OWL 2 DL \cite{krisnadhi2011owlandrules}. It has exponential complexity (\texttt{NExpTime}) for automated reasoning \cite{horrocks2003complexity}. 
\emph{Parsing} a sequence through a DL reasoner would be too complex compared to an external parser: parsing CFGs has complexity $O(n^3)$ \cite{valiant1975cfgcompl} even in the case of ambiguous grammars \cite{earley1970}. Without the use of inverse properties the produced ontology is in DL $\mathcal{SROEL}$, which is within OWL-EL and solvable in polynomial time \cite{horrocks2003complexity}. 
To mitigate the overall complexity we propose a hybrid approach combining CFG parsers and OWL reasoning, to perform sequence classification.

\subsection{Combining CFG parser with OWL-based reasoning} \label{sec:reason-cfg-owl}
We claim that converting CFG in CFG-DL, besides being an interesting theoretical approach, constitutes a relevant pragmatic approach to perform automatic sequence classification that can benefit from an explicit knowledge representation, using OWL ontologies. 
In practice, given a CFG $G$, after recognising a sequence $s \in L(G)$ using a parser for $G$, the resulting parse tree can be converted, using Algorithm \ref{algo:parse-tree-in-owl} to instantiate the OWL ontology $O$ resulting from Algorithm \ref{algo:cfg-in-owl}, for $G$. 

\begin{algorithm}[ht]
    \caption{Parse tree in OWL}
    \label{algo:parse-tree-in-owl}
    \begin{algorithmic}
        \Require $G = (V, \Sigma, R, S)$ in CNF
        \Require $s \subseteq \Sigma^*$ with $s \in L(G)$
        \Require $T$ the parse tree obtained by parsing $s$ with $G$
        \Ensure $T$ is a binary tree

        \ForAll{leaf $l$ in $R$}
            \State \textit{with $\overline{C}_l$ being the concept of the terminal $l$}
            \ForAll{ancestor $a$ of $l$}
                \State \textit{with $\overline{C}_a$ being the concept of the variable $a$}
                \State \texttt{$\blacktriangleright$ Class: $C_l$ SubClassOf: $\overline{C}_a$}
            \EndFor
        \EndFor
    \end{algorithmic} 
\end{algorithm}

Algorithm \ref{algo:parse-tree-in-owl} is based on Definitions \ref{def:cfg-parse-tree} and \ref{def:cnf-binary-parse-tree}. 

The sequence $s$ can be now classified by a DL reasoner according to the classes in $O$ - or of any other ontology aligned to $O$. This process is demonstrated in Section \ref{sec:experiments} with a use case in the music domain.

The same approach can be used to convert the parse tree produced by algorithms such as Neural Network based Part of Speech tagging \cite{akbik2018pos,bohnet2018pos} or Constituency Parsing \cite{tian2020dparse,yang2020dparse,mrini2019dparse}.

\subsection{Comparison with $\mathcal{SEQ}$} \label{sec:comparison}
The work presented in \cite{wissmann2012chord} introduces $\mathcal{SEQ}$, an ontology pattern used to model sequence of elements using Description Logic and OWL. 
The method performs sequence classification by identifying sub-sequences through a subsumption relation: the sequence that is being classified subsumes a set of patterns (sub-sequences). Those patterns classify the sequence. 
The author shows how this method is only able to represent \textit{safely-concatenable} CFG. 
A CFG is \textit{safely-concatenable} if its productions are in the form $R \rightarrow t_1 \cdots t_n X$, with $X, R \in V$ and $t_1 \cdots t_n \in \Sigma$ \cite{wissmann2012chord}. $V$ and $\Sigma$ are defined as in Definition \ref{def:cfg}. 
Such restrictions prevent the representation of \textit{self-embedding} grammars \cite{wissmann2012chord}, which are grammars that contain productions of the type $R \rightarrow \alpha R \beta$, with $R \in V$ and $\alpha, \beta \in (V \cup \Sigma)$ \cite{chomsky1959cfgproperties}.
Our proposal overcomes this limitation by directly reflecting the semantics of a production, as shown in Proof \ref{proof:cfg_in_cfgdl} of Theorem \ref{theorem:cfg_in_cfgdl}.

\section{Experiments} \label{sec:experiments}
In this section we apply our approach to the music domain \footnote{The code of the experiments is available at \url{https://github.com/n28div/CFGOwl} under CC-BY License.} to perform the automatic analysis of harmonic progressions. 
Harmonic progressions are defined as sequences of chords, their analysis consists in assessing the underlying function of each chord \cite{pachet2000solarblues}. Traditionally, it is performed by trained musicians since a deep knowledge and understanding of the music domain is required. The correctness depends on the taxonomy used and on the context in which the sequence is analysed (e.g. the genre).

In music theory, harmony is a well-researched area, and several taxonomies have been proposed to perform this task \cite{rohrmeier2011cfg}. Most approaches classify each chord based on its tonal function, according to western musical theory, using CFG \cite{rohrmeier2011cfg, dehaas2013harmtrace} or Probabilistic CFG \cite{harasim2018generalized}.
The implementation of these grammars can be problematic and relies on different techniques, such as Haskell datatypes in \cite{dehaas2013harmtrace} or  extensions to the definition of CFG in \cite{harasim2018generalized}.
In \cite{keller2013improviser} a CFG is used to detect sub-sequences, called \textit{bricks}. \textit{Bricks} are classes of chords sequences. Their combination defines new \textit{bricks}. A similar approach is explored in \cite{pachet2000solarblues}, where the definition of \textit{bricks} (called \textit{idioms}) is performed through the use of a tree-like hierarchical ontology implemented using Object Oriented programming. 

Our experiments are based on a subset of the rules implemented by \cite{keller2013improviser}, which we convert into an OWL ontology using Algorithm \ref{algo:cfg-in-owl}. 

\subsection{Grammar subset} 
\label{sec:exp:grammar-subset}

The CFG defined in \cite{keller2013improviser} can be formalized as $G_k = (V, \Sigma, R, S)$ where the set of variables $V$ is the set of sub-sequences that will be extracted from a harmonic progression, $\Sigma$ is the set of chords, $R$ is the set of productions that maps each sequence to the corresponding set of chords. The starting $S$ can be assigned to a special variable $V_s \in V$ such that $\forall \  t \in \Sigma \  \exists f \in R\ :s.t.\ f(V_s) = t$.
To obtain a more tractable example, we extract a subset of the whole grammar $G_k$: we will only use the variables, terminals, and productions that are sufficient to analyze the tune \textit{Blue Bossa} by \textit{Dorham Kenny}. We then expand the grammar to include a few other productions that should not appear in the final analysis, to investigate how accurately the ontology reflects a grammar-based approach. A correspondence between the analysis of \cite{keller2013improviser} and our results provides empirical evidences of the method correctness.

\begin{figure}[htp]
    \centering
    \begin{alignat*}{2}
        \text{OnOffMinorIV\_Cm} \rightarrow&\ \text{MinorOn\_Cm}\ \text{Off\_F} \\
        \text{MinorOn\_Cm} \rightarrow&\ \text{C:min}\ |\ \text{C:minmaj7}\ |\ \text{C:min6}\ |\ \text{C:min7} \\ 
        \text{Off\_F} \rightarrow&\ \text{F:7}\ |\ \text{F}\ |\ \text{F:maj}\ |\ \text{F:min}\ |\ \text{F:min7}\ |\ \text{F:minmaj7}\ |\ \text{F:dim7} \\ 
        \text{SadCadence\_Cm} \rightarrow&\ \text{SadApproach\_Cm}\ \text{MinorOn\_Cm} \\
                              &|\ \text{SadApproach\_Cm}\ \text{MinorOn\_Cm} \\
                              &|\ \text{F:7(\#11)}\ \text{MinorPerfectCadence\_Cm} \\ 
        \text{SadApproach\_Cm} \rightarrow&\ \text{D:hdim7}\ \text{G:7} \\ 
        \text{MinorPerfectCadence\_Cm} \rightarrow&\ \text{G:7}\ \text{C:min7} \\ 
        \text{StraightCadence\_Db} \rightarrow&\ \text{StraightApproach\_Db}\ \text{Db} \\ 
                                   &|\ \text{StraightApproach\_Db}\ \text{Db:maj7} \\ 
        \text{StraightApproach\_Db} \rightarrow&\ \text{Eb:min7}\ \text{StraightApproach\_C\_0} \\
                                    &|\ \text{Eb:min7}\ \text{Ab:7} \\
        \text{StraightApproach\_C\_0} \rightarrow&\ \text{Ab:7}\ \text{C:7/Bb}
    \end{alignat*}
    \caption{Productions of grammar $G_{k_1} \subseteq G_k$. Only productions are listed. The set of terminals and variables is the one used in the productions.}
    \label{fig:experiment-grammar-1-productions}
\end{figure}

Figure \ref{fig:experiment-grammar-1-productions} shows the formalization of the rules strictly needed to classify \textit{Blue Bossa} as performed in \cite{keller2013improviser}. 
Using algorithm \ref{algo:cfg-in-owl} we convert the grammar in figure \ref{fig:experiment-grammar-1-productions} into OWL. The resulting ontology contains $130$ axioms. 
Using algorithm \ref{algo:seq-in-owl-for-cfgdl} we convert the chord annotations of \textit{Blue Bossa},  taken from \cite{keller2013improviser} into OWL. The resulting ontology contains a total of $29$ axioms. 
By joining the two ontologies, we obtain a final ontology with a total of $159$ axioms. 
The ontology correctly parses the sequence, as can be seen from table \ref{tab:parsing-experiment-1}.

\setlength{\tabcolsep}{12pt}
\renewcommand{\arraystretch}{1.2}
\begin{table}[htp]
\centering
\caption{Parsing results for the harmonic progression of \textit{Blue Bossa} using the grammar of figure \ref{fig:experiment-grammar-1-productions}. The class identified by \cite{keller2013improviser} is represented in bold text}.

\begin{tabularx}{\textwidth}{ c X } 
\toprule
Chord & Inferred classes \\ \midrule
C:min7  & \textbf{OnOffMinorIV\_Cm} VariableOne C:min7 MinorOn\_Cm                                               \\ 
F:min7  & \textbf{OnOffMinorIV\_Cm} VariableTwo F:min7 Off\_F                                                    \\ 
D:hdim7  & VariableOne SadApproach\_Cm \textbf{SadCadence\_Cm} D:hdim7                                            \\ 
G:7  & MinorPerfectCadence\_Cm VariableTwo VariableOne G:7 SadApproach\_Cm \textbf{SadCadence\_Cm}            \\ 
C:minmaj7  & C:minmaj7 VariableTwo MinorOn\_Cm \textbf{SadCadence\_Cm}                                              \\ 
Eb:min7  & VariableOne StraightApproach\_Db \textbf{StraightCadence\_Db} Eb:min7                                  \\ 
Ab:7  & VariableTwo VariableOne StraightApproach\_Db Ab:7 \textbf{StraightCadence\_Db} StraightApproach\_C\_0  \\ 
Db:maj7  & VariableTwo Db:maj7 \textbf{StraightCadence\_Db}                                                       \\ 
D:hdim7  & VariableOne SadApproach\_Cm \textbf{SadCadence\_Cm} D:hdim7                                            \\ 
G:7  & MinorPerfectCadence\_Cm VariableTwo VariableOne G:7 SadApproach\_Cm \textbf{SadCadence\_Cm}            \\ 
C:minmaj7 & C:minmaj7 VariableTwo MinorOn\_Cm \textbf{SadCadence\_Cm} \\ \bottomrule
\end{tabularx}

\label{tab:parsing-experiment-1}
\end{table}

\subsection{Reasoning complexity} \label{sec:exp:reasoning-complexity}
As discussed in section \ref{sec:method:complexity}, the computational complexity of parsing a sequence using a CFG-DL is exponential. 
On figure \ref{fig:time-complexity} we parse the song \textit{Blue Bossa} using the grammar defined in Section \ref{sec:exp:grammar-subset}. We then progressively add random productions to the grammar that do not affect the classification. 
At each iteration we add $5$ new productions, which have a random number of right-hand sides sampled in the range $[1, 10]$. 
Each production is of type $R \rightarrow AB$ $80\%$ of the time and $R \rightarrow t$ $20\%$ of the time, to reflect the higher frequency of $R \rightarrow AB$ productions, especially when a CFG is expressed in CNF.

\begin{figure}[htp]
    \centering
    \begin{subfigure}{0.45\textwidth}
        \centering
        \includegraphics[width=\textwidth]{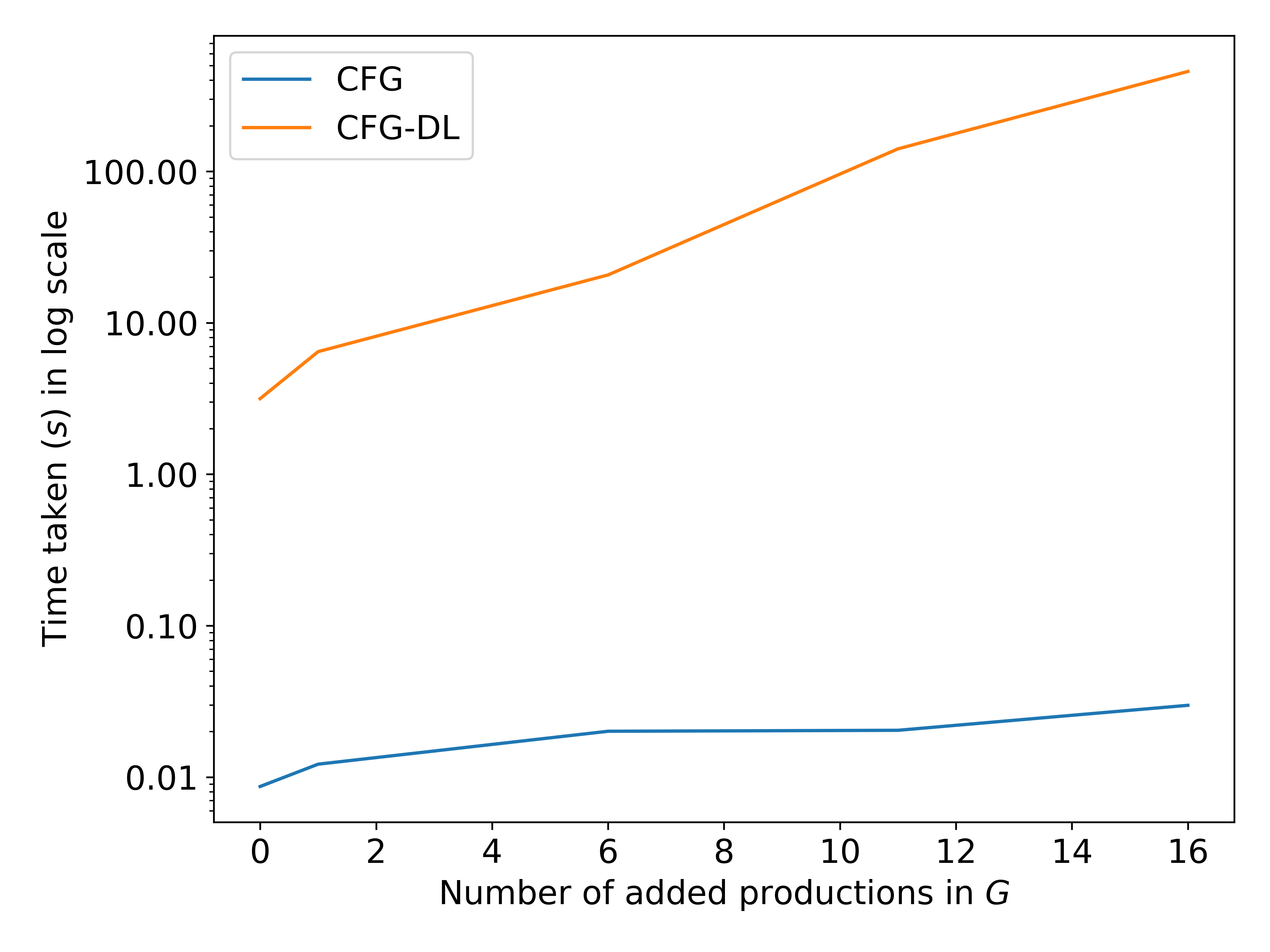}
        \caption{Time taken (y axis, logarithmic) as random productions are added to the grammar (x axis). Parsing using DL is compared to the use of Earley parser on the corresponding CFG. The time taken by using the hybrid approach is $200\%$ ($3$ orders of magnitude) than using DL parsing.}
        \label{fig:time-complexity}
    \end{subfigure}
    \hfill
    \begin{subfigure}{0.45\textwidth}
        \centering
        \includegraphics[width=\textwidth]{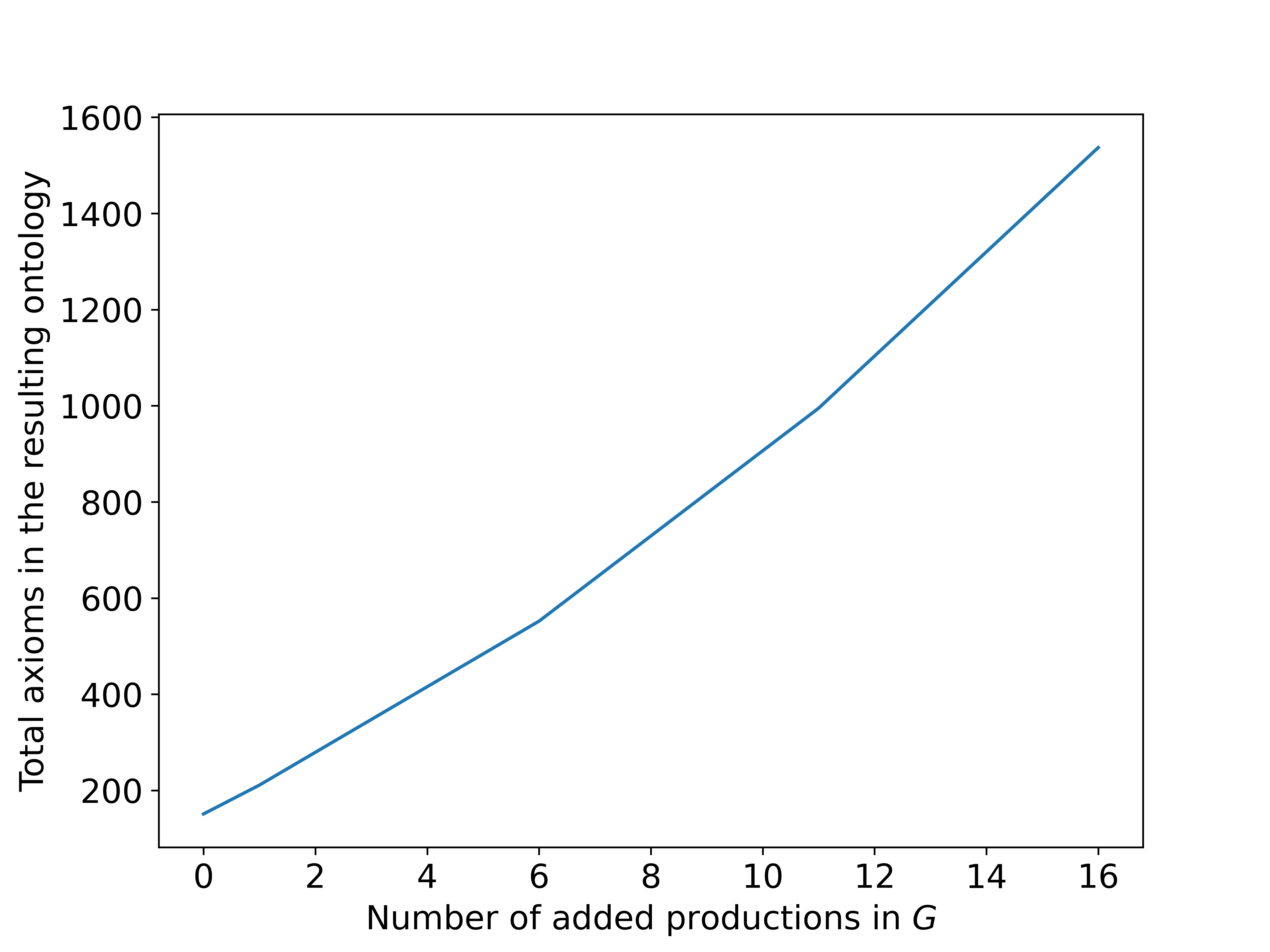}
        \caption{Number of total axioms in the ontology (y axis) as random productions are added to the grammar (x axis). As more productions are added to the grammar, say $N$, roughly $10N$ axioms are inserted in the ontology. Since the computational complexity directly depends on the number of axioms in the ontology, the resulting CFG-DL is inefficient in real-world settings.}
        \label{fig:total-axioms}
    \end{subfigure}
    \caption{Empirical results of the computational complexity when using a CFG-DL to parse the song \textit{Blue Bossa}.}
    \label{fig:empirical-results}
\end{figure}

In Figure \ref{fig:empirical-results} empirical results from the described experiments are shown. 
Figure \ref{fig:time-complexity} shows how as productions are added to the grammar, the time complexity of CFG-DL increases exponentially.
This is a consequence of the proportional increase of axioms as new productions are added (Figure \ref{fig:total-axioms}).
When using the hybrid approach of Section \ref{sec:reason-cfg-owl} the computational complexity is much lower. In Figure \ref{fig:time-complexity} the time required to classify a sequence is significantly lower. All the experiments are executed using the Pellet reasoner \cite{sirin2007pellet} on a $2.4 GHz$ \textit{Intel i5-6300U} CPU and $8 GB$ of RAM under regular computational load. 

Even though the results of the two methods are indistinguishable, it is important to note that if the sequence is modified, Algorithm \ref{algo:parse-tree-in-owl} need to be executed again, while a CFG-DL produced with Algorithm \ref{algo:cfg-in-owl} would be able to classify the new element without any additional effort.
We plan to address this aspect in future works, for instance by combining Algorithm \ref{algo:cfg-in-owl} and Algorithm \ref{algo:parse-tree-in-owl}.

\subsection{Subsequence classification} 
\label{sec:exp:further-classification}

A complete understanding of the CFG is required to interpret Figure \ref{fig:experiment-grammar-1-productions} and the results in Table \ref{tab:parsing-experiment-1}.
To obtain an higher interpretability, it is sufficient to expand the ontology produced by Algorithm \ref{algo:cfg-in-owl} and increase the level of abstraction or by aligning other relevant ontologies.
In the example of Figure \ref{fig:grammar-ontology-ext}, we can align the results with a domain-specific ontology, such as the Music Theory Ontology \cite{rashid2018music}. Differently from \cite{pachet2000solarblues}, the grammar and ontology definitions are decoupled in our approach.

\begin{figure}[ht]
    \centering
    \includegraphics[width=0.6\textwidth]{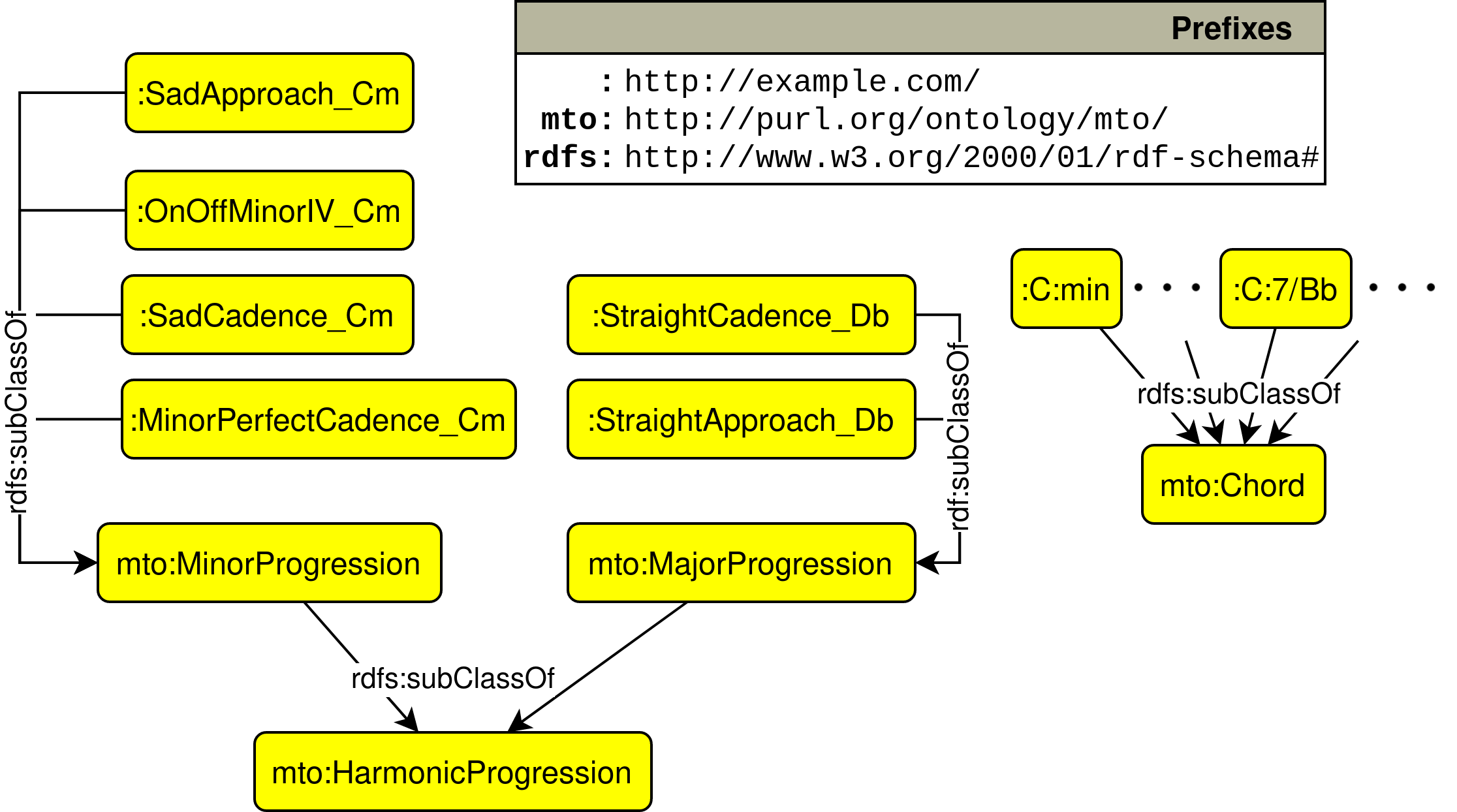}
    \caption{Ontology imported by the ontology generated by algorithm \ref{algo:cfg-in-owl}.}
    \label{fig:grammar-ontology-ext}
\end{figure}

The classification in Table \ref{tab:parsing-experiment-2} is obtained using the ontology of Figure \ref{fig:grammar-ontology-ext}. 
The results are arguably easier to interpret when compared to Table \ref{tab:parsing-experiment-1}, without any update on the original CFG.
By converting relevant grammars using Algorithm \ref{algo:cfg-in-owl}, a Knowledge Graph can be populated using the results of the parsing procedure.
Additional classification can then be performed by aligning additional ontologies.
For instance, to classify modal passages from a major to a minor progression (i.e. chord progression that transition from a major progression to a minor progression) as \textit{X} it would be sufficient to define an axiom such as

{\scriptsize
\begin{align*}
\blacktriangleright & \texttt{((mto:MajorProgression and } \overline{V}_1 \texttt{) or (mto:MinorProgression and } \overline{V}_2 \texttt{)) SubClassOf: X}
\end{align*}
}%

\begin{table}[ht]
    \centering
    \caption{Parsing results for the harmonic progression of \textit{Blue Bossa} using the grammar of Figure \ref{fig:experiment-grammar-1-productions} and importing Music Theory Ontology as shown in Figure \ref{fig:grammar-ontology-ext}.}
    \begin{tabular}{ c c } 
    \toprule
        Chord & Progression type \\ \midrule
        C:min7 & Minor \\
        F:min7 & Minor \\ 
        D:hdim7 & Minor \\ 
        G:7 & Minor \\ 
        C:minmaj7 & Minor \\ 
        Eb:min7 & Major \\ 
        Ab:7 & Major \\ 
        Db:maj & Major \\ 
        D:hdim7 & Minor \\ 
        G:7 & Minor \\ 
        C:minmaj7 & Minor \\ 
    \bottomrule
    \end{tabular}
    \label{tab:parsing-experiment-2}
\end{table}

\section{Conclusions} \label{sec:conclusions}
We present a novel approach to model a Context Free Grammar in Chomsky Normal Form using Description Logic.
The computational complexity, as analysed in Sections \ref{sec:method:complexity} and \ref{sec:exp:reasoning-complexity}, is too high to favour the usage of OWL for parsing sequences.
However, as shown  in Section \ref{sec:exp:further-classification}, it enables the alignment of approaches based on Context Free Grammars with technologies typically used in the Semantic Web. 
Sequences can be represented and classified using OWL in an effective way by combining it with traditional parsing algorithms. This form of classification can be used for tasks such as the computation of similarity between two sequences. The inference of these similarities is of great use in the Music Information Retrieval field, where it is hard to define a similarity metric between two harmonic progression.
The same approach, however, can be applied to other fields where sequences have been modeled using formal grammar, such as natural language processing \cite{lawrence2000nlp}, bio-informatics \cite{damasevicius2010dna} and programming languages \cite{aguair2019ooco}.
The hybrid approach using CFG and OWL ontologies allows a shift in the grammar modeling process: existing extensions, such as Combinatory Categorical Grammars (CCG) \cite{steedman2011ccg}, have been proposed to transparently take into account the semantics of a sequence, along-side the syntactical aspects. It is possible to formalize a CCG in terms of DL, with a similar approach as the one presented in Section \ref{sec:method}, and develop grammars whose semantic information is fueled by an expressive ontology.

\subsubsection*{Acknowledgements}
This project has received funding from the European Union’s Horizon 2020 research and innovation programme under grant agreement No 101004746.

\bibliographystyle{abbrv}
\bibliography{references}

\end{document}